\documentclass[a4paper,11pt]{article}
\pdfoutput=1 % if your are submitting a pdflatex (i.e. if you have
\usepackage{jheppub} % for details on the use of the package, please
\usepackage{amsmath}
\usepackage{amsthm}
\usepackage{physics}
\usepackage{mathtools}

\newtheorem{prop}{Proposition}

\usepackage[T1]{fontenc} % if needed

\title{\boldmath Sandwiched Renyi Relative Entropy in AdS/CFT}

\author[a,b]{Reginald J. Caginalp,}

\affiliation[a]{Berkeley Center for Theoretical Physics, University of California, Berkeley, CA 94720, USA}
\affiliation[b]{Lawrence Berkeley National Laboratory, Berkeley, CA 94720, USA}

\emailAdd{caginalp@berkeley.edu}

\abstract{We explore the role of sandwiched Renyi relative entropy in AdS/CFT and in finite-dimensional models of holographic quantum error correction. In particular, in the context of operator algebra quantum error correction, we discuss a suitable generalization of sandwiched Renyi relative entropy over finite-dimensional von Neumann algebras.  It is then shown that the equality of bulk and boundary sandwiched relative Renyi entropies is equivalent to algebraic encoding of bulk and boundary states, the Ryu-Takayanagi formula, the equality of bulk and boundary relative entropy, and subregion duality. This adds another item to an equivalence theorem between the last four items established in~\cite{HarlowQEC}. We then discuss the sandwiched Renyi relative entropy defined in terms of modular operators, and show that this becomes the definition naturally suited to the finite-dimensional models of holographic quantum error correction. Finally, we explore some numerical calculations of sandwiched Renyi relative entropies for a simple holographic random tensor network in order to obtain a better understanding of corrections to the exact equality of bulk and boundary sandwiched relative Renyi entropy. }

\begin{document} 
\maketitle
\flushbottom

\section{Introduction}

Recent work has uncovered the importance of quantum information in understanding quantum gravity. The anti-de Sitter/conformal field theory (AdS/CFT) correspondence states that any theory of quantum gravity in $(d+1)$-dimensional anti-de Sitter space is equivalent to a conformal field theory in $d$ dimensions~\cite{Maldacena,Gubser,WittenAdSCFT}. The Ryu-Takayagi formula~\cite{RT} and the Hubeney-Rangamani-Takayagi~\cite{HRT} formula show that entanglement entropies in the CFT are equal to areas of minimal or extremal surfaces in the bulk. These formulae can be derived directly from the basic AdS/CFT dictionary~\cite{RTDeriv,HRTDeriv}.  The role of the Renyi entropy, a generalization of entanglement entropy, in AdS/CFT has also been studied~\cite{DongRenyi}.

An important aspect of this connection has been the role that quantum error correction plays in the encoding of the bulk in the boundary~\cite{HarlowQEC,HarlowTASI,ADH,DHW,Happy,Donnelly,DongEntSpec,Akers,JLMS}. In particular, \cite{HarlowQEC} established a theorem demonstrating the equivalence between (i) subregion duality, (ii) the equality of bulk and boundary relative entropy, (iii) algebraic encoding, and (iv) the RT formula.

The purpose of this paper is to discuss the role of sandwiched relative Renyi entropy in holographic quantum error correction. The sandwiched Renyi entropy is a generalization of the relative entropy, which reduces to the usual relative entropy if we take the limit of the Renyi index to 1. In particular, we show that the the equivalence of bulk and boundary relative Renyi entropy is equivalent to the conditions stated above (subregion duality, algebraic encoding, the RT formula, and the equivalence between bulk and boundary relative entropy) in the context of operator algebra quantum error correction with complementary recovery. Along the way, we will define a sandwiched relative Renyi entropy in the context of finite-dimensional von Neumann algebras. We will then show that this definition follows from the definition of relative Renyi entropy on general von Neumann algebras, which is defined from modular operators. Finally, we will discuss results from numerical simulations on a simple holographic tensor model which illustrates the behavior of the sandwiched relative Renyi entropy near a phase transition. These numerical calculation are a model for approximate error correction, in contrast to the theorem we will discuss which concerns an equivalence between exact equality of bulk and boundary sandwiched Renyi relative entropies.

 The sandwiched relative Renyi entropy has previously been used to illuminate the connections between information, entanglement, gravity, and quantum field theory. For example,~\cite{Bao} discussed the bulk dual of the so-called refined relative Renyi entropy.  The quantum null energy condition (QNEC)~\cite{QNEC} can be expressed as a second shape deformation of the relative entropy of a given state and the vacuum. For free quantum field theories, it has been shown that this can be generalized to the sandwiched Renyi relative entropies~\cite{Moosa} for certain values of the Renyi index. Our work in the present paper is closely related to the recent article~\cite{KFR} on corrections to the equality of bulk and boundary relative entropy in AdS/CFT.

We begin this paper with a discussion of the basic definitions and properties of sandwiched relative Renyi entropy as well as a brief review of the theory of holographic quantum error correction. 

\section{Sandwiched Renyi Relative Entropy}

Given a Hilbert space $\mathcal H$ and two density operators $\rho, \sigma$ on the Hilbert space, the sandwiched relative Renyi entropy~\cite{SRD1,SRD2,SRD3} is defined as 
$$S_n ( \rho || \sigma) \equiv \frac{1}{n-1} \log \left [ \text{Tr} \left[ \left ( \sigma^{\frac{1-n}{2n} } \rho \sigma^{\frac{1-n}{2n} }  \right )^n\right] \right ].$$
This is also known in the literature as the sandwiched Renyi divergence (SRD).
\begin{prop}
The sandwiched relative Renyi entropy is invariant under unitary transformations. That is, given a unitary transformation $U$ on the Hilbert space $\mathcal{H}$, we have 
$$S_n ( U \rho U^\dagger ||U  \sigma U^\dagger) =S_n ( \rho || \sigma)$$
for any density matrices $\rho, \sigma$ on $\mathcal{H}$.
\end{prop}
\begin{proof} 
This follows from the definition: 
\begin{multline}
 S_n ( U \rho U^\dagger ||U  \sigma U^\dagger) = \frac{1}{n-1} \log \left [ \text{Tr} \left[ \left ( (U \sigma U^\dagger)^{\frac{1-n}{2n} } U \rho U^\dagger(U \sigma U^\dagger)^{\frac{1-n}{2n} }  \right )^n\right] \right ] \\
 = \frac{1}{n-1} \log \left [ \text{Tr} \left[ \left ( U \sigma ^{\frac{1-n}{2n} } U^\dagger U \rho U^\dagger U \sigma^{\frac{1-n}{2n} } U^\dagger   \right )^n\right] \right ] = \frac{1}{n-1} \log \left [ \text{Tr} \left[ \left ( U \sigma ^{\frac{1-n}{2n} }  \rho  \sigma^{\frac{1-n}{2n} } U^\dagger   \right )^n\right] \right ]\\
 = \frac{1}{n-1} \log \left [ \text{Tr} \left[  U \left (  \sigma ^{\frac{1-n}{2n} }  \rho  \sigma^{\frac{1-n}{2n} } \right )^n U^\dagger \right] \right ] = \frac{1}{n-1} \log \left [ \text{Tr} \left[  \left (  \sigma ^{\frac{1-n}{2n} }  \rho  \sigma^{\frac{1-n}{2n} } \right )^n  \right] \right ] = S_n ( \rho || \sigma).
\end{multline}
\end{proof}

In addition to being invariant under unitary transformations, the sandwiched relative Renyi entropy is always strictly positive, unless the two density matrices are the same, in which case it is zero~\cite{SRD1,SRD2,SRD3}. 

\begin{prop}
Let $\mathcal{H}$ be a Hilbert space, and let $\rho, \sigma$ be two density matrices on $\mathcal{H}$. Then $S_n ( \rho || \sigma) \geq 0$, and $S_n ( \rho || \sigma) = 0$ if and only if $\rho = \sigma$.
\end{prop}

Moreover, the limit as $n \rightarrow 1$ of the relative Renyi entropy is the usual relative entropy~\cite{SRD1,SRD2,SRD3}. 

\begin{prop}
Let $\mathcal{H}$ be a Hilbert space, and let $\rho, \sigma$ be two density matrices on $\mathcal{H}$. Then
$$\lim_{n \rightarrow 1} S_n ( \rho || \sigma ) = S( \rho || \sigma),$$
where 
$$S(\rho || \sigma) \equiv \text{Tr} ( \rho \log \rho) - \text{Tr} ( \rho \log \sigma)$$
is the relative entropy.
\end{prop}

In addition, the sandwiched Renyi relative entropy is monotonic in the Renyi index $n$~\cite{SRD3}. 

\begin{prop}
Let $\mathcal{H}$ be a Hilbert space, and let $\rho, \sigma$ be two density matrices on $\mathcal{H}$, and suppose $n_1,n_2 \in (0, \infty) \backslash \{1\}$ with $n_1 \leq n_2.$ Then 
$$S_{n_1} ( \rho || \sigma) \leq S_{n_2} ( \rho || \sigma). $$
\end{prop}
The sandwiched relative Renyi entropy obeys the data-processing inequality for $n \geq 1/2$. That is, when we apply a quantum channel to the density matrices, the sandwiched relative Renyi entropy decreases or stays the same. See, for example,~\cite{SRD3} and references therein for a discussion. 
\begin{prop} 
Let $\mathcal{H}$ be a Hilbert space, let $\rho, \sigma$ be two density matrices on $\mathcal{H}$, let $n\geq 1/2,$ and let $\Lambda$ be a quantum channel. Then  
$$S_{n} ( \rho || \sigma) \geq S_{n} ( \Lambda[\rho] || \Lambda[\sigma]). $$
\end{prop}

\section{Holographic Quantum Error Correction}

We briefly review some aspects of the connection between quantum error correction and  holography~\cite{HarlowTASI, HarlowQEC}. Consider a finite-dimensional Hilbert space $\mathcal{H}$ that factorizes as $\mathcal{H} = \mathcal{H}_A \otimes \mathcal{H}_{\bar A},$ and a subspace $\mathcal{H}_{code} \subseteq \mathcal{H},$ and a von Neumann algebra $\mathcal{M}$ acting on the code subspace $\mathcal{H}_{code}$. 

Then there exists a decomposition of $\mathcal{H}_{code}$, 
$$\mathcal{H}_{code} = \oplus_\alpha ( \mathcal{H}_{a_\alpha} \otimes \mathcal{H}_{\bar a _\alpha} )$$
such that $\mathcal{M}$ is the set of all operators of the form $\oplus_\alpha ( \mathcal{O}_\alpha \otimes I_{\bar{a}_\alpha})$ where the $\mathcal{O}_\alpha$'s are operators on $\mathcal{H}_{a_\alpha}$. In the above setup, $\mathcal{H}$ is the analog of the full Hilbert space of the conformal field theory, while $\mathcal{H}_{code}$ corresponds to a code subspace, e.g., the set of states that are perturbatively close to some smooth classical bulk geometry. $\mathcal{H}_A$ corresponds to the Hilbert space of some spatial region $A$ in the CFT, while $\mathcal{M}$ represents the set of operators with support in the entanglement wedge of $A$, $\mathcal{E}(A).$

These error correcting codes work by encoding a state on $\mathcal{M}$ in the ``physical" Hilbert space $\mathcal{H}$. In particular, there are states $\chi_\alpha$ on $\mathcal{H}_{\bar a_\alpha}$ and a unitary transfomation $U_A$ on $\mathcal{H}_A$ such that a state $\tilde \rho$ on $\mathcal{H}_{code}$ is mapped to a state with the following density matrix on $\mathcal{H}_A$: 
$$\tilde \rho_A = U_A \left [ \oplus_\alpha ( p_\alpha \rho_{a_\alpha} \otimes \chi_\alpha) \right ] U_A^\dagger.$$ 
Each of the $\rho_{a_\alpha}$'s is defined so that 
$$p_\alpha \rho_{a_\alpha} = \text{Tr}_{\bar{a}_\alpha} \tilde \rho_{\alpha \alpha},$$
 where $\tilde \rho_{\alpha \alpha}$ is the $\alpha$th block of the density matrix of $\mathcal{H}_{code}$. Furthermore, each of the $\rho_{a_\alpha}$'s is normalized so that $\text{Tr} \rho_{a_\alpha} = 1$, and $\sum_\alpha p_\alpha =1$.

The above encoding satisfies an equivalent of the Ryu-Takayanagi formula. In particular, 
$$S(\tilde \rho_A) = \text{Tr}(\tilde \rho_A \mathcal{L}_A ) + S(\rho_a: \mathcal{M}),$$
where $\mathcal{L}_A = \oplus_\alpha S(\chi_\alpha) (I_{a_\alpha} \otimes I_{\bar a _\alpha} )$ is the analog of the area operator, $S(\chi_\alpha)= - Tr( \chi_\alpha \log \chi_\alpha)$ is the von Neumann entropy of $\chi_\alpha$, and $S(\rho_a: \mathcal{M})$ is the algebraic von Neumann entropy over the algebra $\mathcal{M}$. It is defined as 
$$S(\rho_a : \mathcal{M} ) \equiv \sum_\alpha p_\alpha S(\rho_{a_\alpha}) - \sum_\alpha p_\alpha \log (p_\alpha).$$
This entropy consists of an average of the von Neumann entropy of each block, plus a ``classical'' term that is the Shannon entropy of the probability distribution $\{p_\alpha\}$.

This error correcting code also satisfies the equivalence of ``bulk'' and ``boundary'' relative entropy. To see this, consider two states on the code subspace that are encoded in the usual way:
$$\tilde \rho_A = U_A \left [ \oplus_\alpha ( p_\alpha \rho_{a_\alpha} \otimes \chi_\alpha) \right ] U_A^\dagger, \text{		}\tilde \sigma_A = U_A \left [ \oplus_\alpha ( q_\alpha \sigma_{a_\alpha} \otimes \chi_\alpha) \right ] U_A^\dagger.$$
We then compute the relative entropy, using the fact that it is invariant under unitary transformations
$$S(\tilde \rho_A || \tilde \sigma_A) = S(\oplus_\alpha ( p_\alpha \rho_{a_\alpha} \otimes \chi_\alpha) || \oplus_\alpha ( q_\alpha \sigma_{a_\alpha} \otimes \chi_\alpha))$$
In order to calculate this, we will need the logarithm of block diagonal matrices $\oplus_\alpha ( p_\alpha \rho_{a_\alpha} \otimes \chi_\alpha)$. For a block diagonal matrix, $M = \oplus_\alpha M_\alpha$, we have that $M^n = \oplus M_\alpha^n$ and so (by the Taylor series for matrix exponentiation) $\exp(M) = \oplus_\alpha \exp(M_\alpha)$. Therefore, $\log M = \oplus_\alpha \log(M_\alpha).$ Therefore, 
$$\log (\oplus_\alpha ( p_\alpha \rho_{a_\alpha} \otimes \chi_\alpha)) =\oplus_\alpha \log p_\alpha I_{a_\alpha \bar a _\alpha}+ \oplus_\alpha (\log \rho_{a_\alpha} \otimes I_{\bar{a}_\alpha})+\oplus_\alpha (I_{a_\alpha} \otimes \log \chi_{\alpha}).$$
Hence, 
$$Tr ((\oplus_\alpha ( p_\alpha \rho_{a_\alpha} \otimes \chi_\alpha)) \log (\oplus_\alpha ( p_\alpha \rho_{a_\alpha} \otimes \chi_\alpha))) = \sum_\alpha p_\alpha \log p_\alpha + \sum_\alpha p_\alpha Tr(\rho_\alpha \log \rho_{a_\alpha})+ \sum_\alpha p_\alpha Tr(\chi_\alpha \log \chi_\alpha),$$
and, similarly, 
$$Tr ((\oplus_\alpha ( p_\alpha \rho_{a_\alpha} \otimes \chi_\alpha)) \log (\oplus_\alpha ( q_\alpha \sigma_{a_\alpha} \otimes \chi_\alpha))) = \sum_\alpha p_\alpha \log q_\alpha + \sum_\alpha p_\alpha Tr(\rho_\alpha \log \sigma_{a_\alpha})+ \sum_\alpha p_\alpha Tr(\chi_\alpha \log \chi_\alpha).$$
Hence, 
\begin{multline}
S(\tilde \rho_A || \tilde \sigma_A)=Tr ((\oplus_\alpha ( p_\alpha \rho_{a_\alpha} \otimes \chi_\alpha)) \log (\oplus_\alpha ( p_\alpha \rho_{a_\alpha} \otimes \chi_\alpha))) -Tr ((\oplus_\alpha ( p_\alpha \rho_{a_\alpha} \otimes \chi_\alpha)) \log (\oplus_\alpha ( q_\alpha \sigma_{a_\alpha} \otimes \chi_\alpha))) \\
= \sum_\alpha p_\alpha \log \left ( \frac{p_\alpha} {q_\alpha} \right ) + \sum_\alpha p_\alpha S(\rho_{a_\alpha}||\sigma_{a_\alpha}) = S(\rho_a || \sigma_a : \mathcal{M}).
\end{multline}
As with the algebraic entropy considered above, this is the sum of two terms, an averaged quantum relative entropy, weighted by the probability distribution $\{p_\alpha\}$ and a ``classical'' term, which is the relative Shannon entropy of the probability distributions $\{p_\alpha\}$ and $\{q_\alpha\}$. Moreover, this algebraic relative entropy is the exact form one obtains from the theory of modular operators, as we discuss below. 

A closely related concept in this construction is that of subregion duality. That is, for all operators on $\mathcal{O} \in M$, there are operators $\mathcal{O}_A$ which acts nontrivially only on $\mathcal{H}_A$ such that $\mathcal{O} \ket \psi = \mathcal{O}_A \ket \psi$ for all $\ket \psi \in \mathcal{H}_{code}$.

It was shown in~\cite{HarlowQEC} that the existence of an encoding map, the RT formula, subregion duality, and the equivalence of bulk and boundary relative entropy are all in fact equivalent for these kinds of finite-dimensional Hilbert spaces.

\section{Sandwiched Renyi Relative Entropy in Holographic Quantum Error Correction}

 %We now calculate the sandwiched relative Renyi entropy for these error-correcting codes.
The so-called $\alpha$-block decomposition described above has received considerable interest in recent years~\cite{DongEntSpec,Akers,Donnelly} and is now relatively well-understood. It is our aim in this section to discuss the sandwiched Renyi relative entropy in this quantum error-correction context and to provide a definition of sandwiched Renyi relative entropy that is suitable for this $\alpha$-block setting. In the next section, we will show that this definition also follows from a modular-theoretic definition of sandwiched Renyi relative entropy.
%We now consider the sandwiched Renyi relative entropy in the context of holographic error correcting codes. 

In particular, we will show that the equality of sandwiched Renyi relative entropy is equivalent to the four statements above. Consider the same setup as in the previous section, and consider two states $\tilde \rho$ and $\tilde \sigma$ on the code subspace encoded in the usual way 
$$\tilde \rho_A = U_A \left [ \oplus_\alpha ( p_\alpha \rho_{a_\alpha} \otimes \chi_\alpha) \right ] U_A^\dagger, \text{		}\tilde \sigma_A = U_A \left [ \oplus_\alpha ( q_\alpha \sigma_{a_\alpha} \otimes \chi_\alpha) \right ] U_A^\dagger.$$
We calculate 
\begin{multline}
S_n(\tilde \rho_A || \tilde \sigma_A) = S_n(\oplus_\alpha ( p_\alpha \rho_{a_\alpha} \otimes \chi_\alpha) || \oplus_\alpha ( q_\alpha \sigma_{a_\alpha} \otimes \chi_\alpha)) \\
= \frac{1}{n-1} \log \left [ \text{Tr} \left[ \left ( (\oplus_\alpha ( q_\alpha \sigma_{a_\alpha} \otimes \chi_\alpha))^{\frac{1-n}{2n} } (\oplus_\alpha( p_\alpha \rho_{a_\alpha} \otimes \chi_\alpha) ) (\oplus_\alpha ( q_\alpha \sigma_{a_\alpha} \otimes \chi_\alpha))^{\frac{1-n}{2n} }  \right )^n\right] \right ] \\
=\frac{1}{n-1} \log \left [ \sum_\alpha p_\alpha^n q_\alpha^{1-n} \text{Tr} \left[ \left ( \sigma_{a_\alpha}^{\frac{1-n}{2n} } \rho_{a_\alpha} \sigma_{a_\alpha}^{\frac{1-n}{2n} }  \right )^n \otimes \chi_\alpha \right ] \right ]\\
=\frac{1}{n-1} \log \left [ \sum_\alpha p_\alpha^n q_\alpha^{1-n} \text{Tr} \left[ \left ( \sigma_{a_\alpha}^{\frac{1-n}{2n} } \rho_{a_\alpha} \sigma_{a_\alpha}^{\frac{1-n}{2n} }  \right )^n \right ] \cdot \text{Tr}(\chi_\alpha) \right ]\\
=\frac{1}{n-1} \log \left [ \sum_\alpha p_\alpha^n q_\alpha^{1-n} \text{Tr} \left[ \left ( \sigma_{a_\alpha}^{\frac{1-n}{2n} } \rho_{a_\alpha} \sigma_{a_\alpha}^{\frac{1-n}{2n} }  \right )^n \right ] \right ],
\end{multline}
where in the first step, we have used the invariance of sandwiched relative Renyi entropy under unitary transformations, as discussed above. Note that we can write this as 
$$S_n(\tilde \rho_A || \tilde \sigma_A) = \frac{1}{n-1} \log \left [ \sum_\alpha p_\alpha^n q_\alpha^{1-n} \exp [ (n-1) S_n( \rho_{a_\alpha} || \sigma_{a_\alpha})] \right ].$$
As we discuss below, this is exactly the same as the Renyi relative entropy derived for such finite-dimensional von Neumann algebras using modular operators. (Similarly, the modular-theoretic definition of relative entropy reduces to the one defined above for finite-dimensional von Neumann algebras.) Note that this algebraic entropy is not of the form of a classical term plus a weighted average of the quantum entropies of each $\alpha$-block. However, note that the classical relative Renyi entropy of two probability distributions $\{p_\alpha\}$ and $\{q_\alpha\}$ is given by 
$$S_n ( \{p_\alpha\}|| \{q_\alpha\} ) = \frac{1}{n-1} \log \left [ \sum_\alpha p_\alpha^n q_\alpha^{1-n} \right ].$$
This means that when the quantum states are purely classical probability distributions, our algebraic sandwiched relative Renyi entropy reduces to the classical relative Renyi entropy. Moreover, it is clear that when we only have one $\alpha$-block, the algebraic sandwiched relative Renyi entropy reduces to the usual sandwiched relative Renyi entropy defined above. This is exactly as we would expect, in analogy with the corresponding special cases for the algebraic entropy and relative entropy described above.

The classical relative Renyi entropy described above is always greater than or equal to zero, and it is zero if and only if the probability distributions are identical, $p_\alpha = q_\alpha$. We claim that our algebraic relative Renyi entropy is greater than or equal to zero, and that it is zero if and only if the states are the same. To see this, consider the case where $n>1$. $S_n( \rho_{a_\alpha} || \sigma_{a_\alpha}) \geq 0$, so 
\begin{multline}
S_n(\tilde \rho_a || \tilde \sigma_a : \mathcal{M}) = \frac{1}{n-1} \log \left [ \sum_\alpha p_\alpha^n q_\alpha^{1-n} \exp [ (n-1) S_n( \rho_{a_\alpha} || \sigma_{a_\alpha})] \right ] \\
\geq \frac{1}{n-1} \log \left [ \sum_\alpha p_\alpha^n q_\alpha^{1-n}  \right ] = S_n ( \{p_\alpha\}|| \{q_\alpha\} ).
\end{multline}
since $\log$ is a monotone increasing function. Therefore, $S_n(\tilde \rho_a|| \tilde \sigma_a: \mathcal{M})\geq 0$. Moreover, the inequality in the above expression is saturated (i.e., $S_n(\tilde \rho_a || \tilde \sigma_a : \mathcal{M})=S_n ( \{p_\alpha\}|| \{q_\alpha\} )$) if and only if $\rho_{a_\alpha} = \sigma_{a_\alpha}$ for all $\alpha$. Also, as discussed above, $S_n ( \{p_\alpha\}|| \{q_\alpha\} )=0$ if and only if the probability distributions are identical, $p_\alpha = q_\alpha$. Therefore, $S_n(\tilde \rho_a || \tilde \sigma_a: \mathcal{M}) =0$ if and only if $\rho_{a_\alpha} = \sigma_{a_\alpha}$ and $p_\alpha = q_\alpha$ for all $\alpha$, i.e., the states are identical, as claimed. The case with $n<1$ is similar. 

In addition, we claim that our algebraic sandwiched Renyi relative entropy limits to the algebraic relative entropy defined above as $n \rightarrow 1$, just as the sandwiched Renyi relative entropy converges to the relative entropy when $n\rightarrow 1$. 
Note that 
$$\lim_{n \rightarrow 1} S_n( \rho_{a_\alpha} || \sigma_{a_\alpha}) = S( \rho_{a_\alpha} || \sigma_{a_\alpha})$$
for all $\alpha.$ We have 
$$\lim_{n \rightarrow 1}  S_n(\tilde \rho_a || \tilde \sigma_a : \mathcal{M}) = \lim_{n \rightarrow 1}  \frac{1}{n-1} \log \left [ \sum_\alpha p_\alpha^n q_\alpha^{1-n} \exp [ (n-1) S_n( \rho_{a_\alpha} || \sigma_{a_\alpha})] \right ],$$
and 
$$ \lim_{n \rightarrow 1}  \log \left [ \sum_\alpha p_\alpha^n q_\alpha^{1-n} \exp [ (n-1) S_n( \rho_{a_\alpha} || \sigma_{a_\alpha})] \right ] = \log [ \sum_\alpha p_\alpha]  =0. $$
Therefore, by L'Hospital's rule, 
\begin{multline}
\lim_{n \rightarrow 1}  S_n(\tilde \rho_a || \tilde \sigma_a : \mathcal{M}) = \lim_{n \rightarrow 1}  \frac{1}{n-1} \log \left [ \sum_\alpha \exp [ n \log p_\alpha - (n-1) \log q_\alpha +(n-1) S_n( \rho_{a_\alpha} || \sigma_{a_\alpha})] \right ] \\
=\lim_{n \rightarrow 1}  \frac{\sum_\alpha \exp [ n \log p_\alpha - (n-1) \log q_\alpha +(n-1) S_n( \rho_{a_\alpha} || \sigma_{a_\alpha})] ( \log p_\alpha - \log q_\alpha + S_n ( \rho_{a_\alpha} || \sigma_{a_\alpha}) ) }  {\sum_\alpha \exp [ n \log p_\alpha - (n-1) \log q_\alpha +(n-1) S_n( \rho_{a_\alpha} || \sigma_{a_\alpha})]}  \\ 
+ \lim_{n \rightarrow 1}  \frac{\sum_\alpha \exp [ n \log p_\alpha - (n-1) \log q_\alpha +(n-1) S_n( \rho_{a_\alpha} || \sigma_{a_\alpha})] (n-1) \partial_n S_n ( \rho_{a_\alpha} || \sigma_{a_\alpha}) }  {\sum_\alpha \exp [ n \log p_\alpha - (n-1) \log q_\alpha +(n-1) S_n( \rho_{a_\alpha} || \sigma_{a_\alpha})]}  \\
= \frac{\sum_\alpha \exp [ \log p_\alpha ] ( \log p_\alpha - \log q_\alpha + S ( \rho_{a_\alpha} || \sigma_{a_\alpha}) ) }  {\sum_\alpha \exp [  \log p_\alpha ]} = \sum_\alpha p_\alpha \log \left ( \frac{p_\alpha}{q_\alpha} \right ) + \sum_\alpha p_\alpha S ( \rho_{a_\alpha} || \sigma_{a_\alpha}),
\end{multline}
which is the algebraic relative entropy, $S(\rho_a || \sigma_a: \mathcal{M})$, as claimed.

Now, as noted above, it was previously established that subregion duality, algebraic encoding, equality of bulk and boundary relative entropy, and the RT formula are all equivalent for this setting. We have shown that the algebraic encoding implies the equality of bulk and boundary sandwiched relative Renyi entropy (using the algebraic definition discussed above). Equality of bulk and boundary sandwiched relative Renyi entropy implies the equality of bulk and boundary relative entropy, by taking the limit $n \rightarrow 1$, since the algebraic sandwiched relative Renyi entropy converges to the algebraic relative entropy in this limit. Therefore, the equality of bulk and boundary sandwiched relative Renyi entropy is also equivalent to subregion duality, algebraic encoding, equality of bulk and boundary relative entropy, and the RT formula.

\section{Sandwiched Renyi Relative Entropy using Modular Operators}
We now discuss relative entropies and relative Renyi entropies using the theory of modular operators, closely following~\cite{Witten, Lashkari}. We begin by reviewing modular operators. 

Consider a Hilbert space $\mathcal{H}$ that has the following structure,
$$\mathcal{H} = \oplus_{A} (\mathcal{H}_A \otimes \mathcal{H}_{\bar A}),$$
and for simplicity we consider the case where $\dim \mathcal{H}_A = \dim \mathcal{H}_{\bar A}.$ Consider two (normalized) states, $\ket \Psi, \ket \Phi \in \mathcal{H}.$ We can write 
$$\ket \Psi = \sum_A r_A \ket {\psi_A},$$
where each $\ket{ \psi_A}$ is normalized, and $\sum_A |r_A|^2 = 1.$ We can write each $\ket {\psi_A}$ as 
$$\ket {\psi_A} = \sum_i c(i,A) \ket{i,i,A},$$
where $\ket{i,i,A} = \ket{i,A} \otimes \ket{i,A}'$, and $\{ \ket{i,A} \}$ is an orthonormal basis for $\mathcal{H}_A$, and $\{ \ket{i,A}' \}$ is an orthonormal basis for $\mathcal{H}_{\bar{A}}$. Each $\ket {\psi_A}$ is normalized, so $\sum_i |c(i,A)|^2 = 1.$ In the exact same way, we can write 
$$\ket \Phi = \sum_A s_A \ket {\phi_A}, \ket {\phi_A} = \sum_\alpha d(\alpha,A) \ket{\alpha,\alpha,A},$$
where $\ket{\alpha,\alpha,A} = \ket{\alpha,A} \otimes \ket{\alpha,A}'$, and $\{ \ket{\alpha,A} \}$ is an orthonormal basis for $\mathcal{H}_A$, and $\{ \ket{\alpha,A}' \}$ is an orthonormal basis for $\mathcal{H}_{\bar{A}}$. Once again, we have the normalization conditions
$$\sum_A |s_A|^2 = 1, \sum_\alpha |d(\alpha,A)|^2 = 1.$$ Consider the algebra of operators $\mathcal{A}$ defined by operators on $\mathcal{H}$ that are block-diagonal and have the form $\oplus_A  (\mathcal{O}_A \otimes I_{\bar{A}})$, where $I_{\bar{A}}$ is the identity operator on $\mathcal{H}_{\bar{A}}$. We first determine the relative Tomita operator, $S_{\Psi || \Phi}$, which is defined by
$$S_{\Psi || \Phi} \mathcal{O} \ket \Psi =  \mathcal{O}^\dagger \ket \Phi \text{	for all  } \mathcal{O} \in \mathcal{A}.$$
Consider an operator $\mathcal{O} \in \mathcal{A}$, which we can write as $\oplus_A  (\mathcal{O}_A \otimes I_{\bar{A}})$, that acts as follows: 
$$\mathcal{O}\ket{i,A} = \ket{\alpha, A}$$
and 
$$\mathcal{O}_A \ket{j,B} = 0$$
for all $j \neq i$, and for all $B \neq A$. Then the adjoint acts as
$$\mathcal{O}^\dagger \ket{\alpha,A} = \ket{i, A}$$
and 
$$\mathcal{O}^\dagger \ket{\beta,B} = 0$$
for all $\beta \neq \alpha$, and for all $B \neq A$. Therefore, 
$$\mathcal{O} \ket \Psi =r_A c(i,A) \ket{\alpha,i,A},  \mathcal{O}^\dagger \ket \Phi =s_A d(\alpha,A) \ket{i,\alpha,A}.$$
Hence, the relative Tomita operator acts as 
$$S_{\Psi || \Phi} \ket {\alpha, i, A} = \frac{s_A}{r_A} \frac{d(\alpha,A)}{c(i,A)}\ket{i, \alpha, A}.$$
The relative Tomita operator is an anti-linear operator, so the adjoint acts as 
$$S_{\Psi || \Phi}^\dagger \ket { i,\alpha, A} = \frac{s_A^*}{r_A^*} \frac{d(\alpha,A)^*}{c(i,A)^*}\ket{ \alpha,i,  A}.$$
Using this, we can compute the relative modular operator, $\Delta_{\Psi || \Phi} = S_{\Psi || \Phi}^\dagger S_{\Psi || \Phi}.$ We find 
$$\Delta_{\Psi || \Phi} \ket {\alpha, i, A} = \frac{|s_A|^2}{|r_A|^2} \frac{|d(\alpha,A)|^2}{|c(i,A)|^2}\ket{\alpha,i, A}.$$
Define $ \sigma_A$ to be the reduced density matrix of $\ket {\phi_A}$ on $\mathcal{H}_{A}$, and $\bar \rho_A$ to be the reduced density matrix of $\ket {\psi_A}$ on $\mathcal{H}_{\bar A},$ and define $p_A \equiv |r_A|^2, q_A \equiv |s_A|^2$. With these definitions, and the expressions above for $\ket{\psi_A}$ and $\ket{\phi_A}$, we obtain 
$$\Delta_{\Psi || \Phi} = \oplus_A \left (\frac{q_A}{p_A} \sigma_A \otimes \bar \rho_A^{-1} \right ).$$
Having computed the relative modular operator, we now turn to the calculation of algebraic relative entropy.
Recall that the Araki definition of relative entropy over our von Neumann algebra $\mathcal{A}$ is given by
$$S(\Psi || \Phi: \mathcal{A}) \equiv - \bra{\Psi} \log \Delta_{\Psi || \Phi: \mathcal{A}} \ket{\Psi}.$$
Now, for a block diagonal matrix, $M = \oplus_A M_A$, $M^n = \oplus_A M_A^n$ so that (by the Taylor series definition of matrix exponentiation) $\exp M = \oplus_A \exp M_A.$ which means that for a block-diagonal matrix $M = \oplus_A M_A$, $\log M = \oplus_A \log M_A.$ Therefore, 
\begin{multline}
S(\Psi || \Phi: \mathcal{A}) = - \bra{\Psi} \oplus_A \log  \left (\frac{q_A}{p_A} \sigma_A \otimes \bar \rho_A^{-1} \right ) \ket{\Psi}=- \bra{\Psi} \oplus_A \log  \left (\frac{q_A}{p_A} \sigma_A \otimes \bar \rho_A^{-1} \right ) \ket{\Psi} \\
=- \bra{\Psi} \oplus_A \left [ \log  \left (\frac{q_A}{p_A} \right ) I_A \otimes  I_{\bar A} \right ]\ket{\Psi}  - \bra{\Psi} \oplus_A \log  \left (\sigma_A \otimes \bar \rho_A^{-1} \right ) \ket{\Psi}.
\end{multline}
Now, 
$$\log  \left (\sigma_A \otimes \bar \rho_A^{-1} \right ) = \log  \sigma_A \otimes \bar I_{\bar A} - I_A \otimes  \log \bar\rho_A,$$
and the first term in $S(\Psi || \Phi: \mathcal{A})$ becomes 
\begin{multline}- \bra{\Psi} \oplus_A \left [ \log  \left (\frac{q_A}{p_A} \right ) I_A \otimes  I_{\bar A} \right ] \ket{\Psi} = - \sum_A \bra{\psi_A} r_A^* r_A \log \left ( \frac{q_A}{p_A} \right ) \ket {\psi_A}\\
= - \sum_A p_A \log \left ( \frac{q_A}{p_A} \right) = \sum_A p_A \log \left ( \frac{p_A}{q_A} \right).  
\end{multline}
Therefore, the relative entropy becomes:
\begin{multline}
S(\Psi || \Phi: \mathcal{A}) = \sum_A p_A \log \left ( \frac{p_A}{q_A} \right) -\bra{\Psi} \oplus_A \log  \left (\sigma_A \otimes \bar \rho_A^{-1} \right ) \ket{\Psi} \\ 
=\sum_A p_A \log \left ( \frac{p_A}{q_A} \right) + \bra{\Psi} \oplus_A  (I_A \otimes  \log \bar\rho_A) \ket{\Psi}-\bra{\Psi} \oplus_A ( \log \sigma_A \otimes  I_{ \bar A}) \ket{\Psi} \\
=\sum_A p_A \log \left ( \frac{p_A}{q_A} \right) + \sum_A p_A \bra{\psi_A}  I_A \otimes  \log \bar\rho_A\ket{\psi_A}-\sum_A p_A \bra{\psi_A}   \log \sigma_A \otimes  I_{ \bar A} \ket{\psi_A} \\
=\sum_A p_A \log \left ( \frac{p_A}{q_A} \right) + \sum_A p_A \text{Tr} ( \bar \rho_A \log \bar\rho_A)-\sum_A p_A \text{Tr} ( \rho_A   \log \sigma_A ).
\end{multline}
Now, each $\ket{\psi_A}$ is a pure state for each $A$, so $\text{Tr} ( \bar \rho_A \log \bar\rho_A) = \text{Tr} ( \rho_A \log \rho_A)$ for each $A$, where $\rho_A$ is of course the reduced density matrix of $\ket{\psi_A}$ on the Hilbert space $\mathcal{H}_A$. Thus, 
\begin{multline}
S(\Psi || \Phi: \mathcal{A}) =\sum_A p_A \log \left ( \frac{p_A}{q_A} \right) + \sum_A p_A \text{Tr} (  \rho_A \log \rho_A)-\sum_A p_A \text{Tr} ( \rho_A   \log \sigma_A ) \\
 = \sum_A p_A \log \left ( \frac{p_A}{q_A} \right) + \sum_A p_A \left [ \text{Tr} (  \rho_A \log \rho_A)- \text{Tr} ( \rho_A   \log \sigma_A ) \right ].
\end{multline}
This is a sum of a purely classical ``relative entropy'' of the probability distributions $\{p_A\}$ and $\{q_A\}$, and a weighted sum of the quantum relative entropies of the density matrices $\rho_A$ and $\sigma_A$, weighed by the probabilities $p_A$. This is the usual definition of algebraic relative entropy for this type of von Neumann algebra. Moreover, the above expression is the type of relative entropy that appears in the equality of bulk and boundary relative entropies in the context of holographic operator algebra quantum error correction.

The algebraic Renyi relative entropy is defined by Lashkari~\cite{Lashkari} using modular operators as 
$$S_\alpha( \Phi || \Omega : \mathcal{A} ) = \frac{1}{\alpha} \sup_{\Psi \in \mathcal{H}} \log \bra{ \Phi} \Delta_{\Omega || \Psi}^{\alpha} \ket{\Phi},$$
for $\alpha>0$, and similarly for $\alpha<0$. Note that the index $\alpha$ is different than the index $n$ above -- they are related by $\alpha = \frac{n-1}{n}.$ In order to analyze this, we first need to discuss various matrix norms. For a matrix $X$, its $p$-norm is defined to be 
$$\norm{X}_p \equiv Tr(|X|^p)^{1/p}.$$
This norm satisfies Holder's inequality
$$\norm{XY}_1\leq \norm{X}_p \norm{Y}_q,$$
where $p$ and $q$ satisfy $\frac{1}{p}+\frac{1}{q} =1$, $p,q>1$.
Now, $Tr(XC) \leq \norm{XC}_1$ so Holder's inequality tells us
$$ \sup_{\norm{C}_q = 1} Tr(XC) \leq \norm{XC}_1 \leq \norm{X}_p \norm{C}_q = \norm{X}_p,$$
where $\frac{1}{p}+\frac{1}{q} = 1$.
The matrix $X$ has polar decomposition $X=U |X|$. Consider an operator $C_0=A |X|^{p/q} U^\dagger,$ where $A$ is a constant. We have that 
$$\norm{C_0}_q = A \cdot Tr(|X|^p)^{1/q}.$$
By requiring that $\norm{C_0}_q = 1$, we find $A = \frac{1}{Tr(|X|^p)^{1/q}}$ so that 
$$C_0 = \frac{|X|^{p/q}U^\dagger}{Tr(|X|^p)^{1/q}}.$$
We then find 
$$Tr(XC_0) =\frac{Tr(U |X| |X|^{p/q} U^\dagger )}{Tr(|X|^p)^{1/q}}=\frac{Tr(|X|^{1+p/q} )}{Tr(|X|^p)^{1/q}}.$$
Now, $\frac{1}{p}+\frac{1}{q} = 1,$ so 
$$Tr(XC_0) =\frac{Tr(|X|^{p} )}{Tr(|X|^p)^{1/q}} = Tr(|X|^p)^{1-1/q} = Tr(|X|^p)^{1/p} = \norm{X}_p.$$
We know that $ \sup_{\norm{C}_q = 1} Tr(XC) \leq  \norm{X}_p,$ and that $C_0$ saturates this inequality with $\norm{C_0}_q=1$, so we conclude that
$$ \sup_{\norm{C}_q = 1} Tr(XC) =  \norm{X}_p.$$
Now, recall that we wrote our state $\ket \Phi$ as $\ket \Phi = \sum_A s_A \ket{\phi_A}$. Each of the states $\ket{\phi_A}$ can in turn be written as $\ket{\phi_A} = \sum_{\alpha} d(\alpha,A) \ket{\alpha, \alpha, A}.$ The reduced density matrix of $\ket {\phi_A}$ on $\mathcal{H}_A$ is then 
$$\sigma_A = \sum_\alpha |d(\alpha,A)|^2 \ket{\alpha,A} \bra{\alpha, A}.$$
Therefore, for full rank $\sigma_A$, 
$$\ket{\phi_A} = (\sigma_A^{1/2} \otimes I_{\bar A} )\sum_\alpha \ket{\alpha, \alpha, A}.$$
Now, let $X$ and $Y$ be two Hermitian operators, so that 
\begin{multline}\bra{\phi_A} X \otimes Y \ket{\phi_A} = \sum_{\alpha, \beta} \bra{\alpha, \alpha, A}\sigma_A^{1/2} X \sigma_A^{1/2} \otimes Y \ket{\beta, \beta, A} = \sum_{\alpha, \beta} \bra{\alpha, A}\sigma_A^{1/2} X \sigma_A^{1/2} \ket{\beta, A} \bra{\beta,A}  Y \ket{\alpha, A} \\
=  \sum_{\alpha, \beta} \bra{\alpha, A}\sigma_A^{1/2} X \sigma_A^{1/2} Y \ket{\alpha, A} = \text{Tr}_{\mathcal{H}_A} (\sigma_A^{1/2} X \sigma_A^{1/2} Y).
\end{multline}
We can apply this result to the algebraic relative Renyi entropy. We use the same notation for the expansions of $\ket \Psi$ and $\ket \Phi$ as before. We can write the state $\ket \Omega$ as 
$$\ket \Omega = \sum_A t_A \ket{ \omega_A},$$
and write $\ket{\omega_A}$ as 
$$\ket{\omega_A} = \sum_\mu g(\mu,A) \ket{\mu,\mu, A},$$
where $\ket{\mu,\mu,A} = \ket{\mu,A} \otimes \ket{\mu,A}'$, $\{\ket{\mu,A}\}$ is an orthonormal basis for $\mathcal{H}_{A}$, and $\{\ket{\mu,A}'\}$ is an orthonormal basis for $\mathcal{H}_{\bar A}$. We have our usual normalization conditions 
$$\sum_A |t_A|^2 =1, \sum_\mu |g(\mu,A)|^2=1.$$
We write $\tau_A$ for the reduced density matrix of $\ket{\omega_A}$ on $\mathcal{H}_A$, and $\bar \tau_A$ for the reduced density matrix of $\ket{\omega_A}$ on $\mathcal{H}_{\bar A}.$

Recall that our modular operator $\Delta_{\Omega || \Psi}$ is given by 
$$\Delta_{\Omega || \Psi} = \oplus_A \left (\frac{p_A}{w_A} \rho_A \otimes \bar \tau_A^{-1} \right ),$$
where $w_A \equiv |t_A|^2$.
The algebraic relative Renyi entropy is then given by 
$$S_\alpha( \Phi || \Omega : \mathcal{A} ) = \frac{1}{\alpha} \sup_{\Psi | \langle \Psi | \Psi \rangle =1} \log \bra{ \Phi} \Delta_{\Omega || \Psi} \ket{\Phi}= \frac{1}{\alpha} \sup_{\Psi | \langle \Psi | \Psi \rangle =1} \log \left( \sum_A q_A \frac{p_A^\alpha}{w_A^\alpha} \bra{ \phi_A} \rho_A^{\alpha} \otimes \bar \tau_A^{-\alpha}  \ket{\phi_A} \right).$$
We can use our result above to write this as 
\begin{multline}
S_\alpha( \Phi || \Omega : \mathcal{A} ) =  \frac{1}{\alpha} \sup_{\Psi | \langle \Psi | \Psi \rangle =1} \log \left( \sum_A q_A \frac{p_A^\alpha}{w_A^\alpha} \text{Tr}_{\mathcal{H}_A} (\sigma_A^{1/2} \rho_A^{\alpha} \sigma_A^{1/2}  \tau_A^{-\alpha}  ) \right)\\
=  \frac{1}{\alpha} \sup_{\Psi | \langle \Psi | \Psi \rangle =1} \log \left( \sum_A q_A \frac{p_A^\alpha}{w_A^\alpha} \text{Tr}_{\mathcal{H}_A} ( \rho_A^{\alpha} \sigma_A^{1/2}  \tau_A^{-\alpha} \sigma_A^{1/2} ) \right).
\end{multline}
We are taking the supremum over all $\ket \Psi \in \mathcal{H}$ that are normalized. This is equivalent to, in the notation used above, $\sum_A p_A = 1$ and $Tr(\rho_A)=1$ for all $A$. So, we need to take the supremum over all probability distributions $\{p_A\}$ (which we write as $\sup_{\{p_A\}}$, where it is understood that we are taking the supremum over all $p_A$ with all $p_A>0$ and $\sum_A p_A = 1$), and the supremum over all normalized density matrices $\rho_A$.   Define $\eta_A = \rho_A^\alpha$, so that $Tr(\rho_A) = 1$ is equivalent to $\norm{\eta}_{1/\alpha} = 1$. Putting this all together, 
$$S_\alpha( \Phi || \Omega : \mathcal{A} ) =  \frac{1}{\alpha} \sup_{\{p_A\}} \log \left( \sum_A q_A \frac{p_A^\alpha}{w_A^\alpha} \sup_{\eta_A | \norm{\eta_A}_{1/\alpha}=1} \text{Tr}_{\mathcal{H}_A} ( \eta_A \sigma_A^{1/2} \tau_A^{-\alpha} \sigma_A^{1/2} ) \right).$$
Now define $X_A= \tau_A ^{-\alpha/2} \sigma_A^{1/2} $ so that (using the equality of sup-norm and matrix norm derived above)
\begin{multline}
\sup_{\eta_A | \norm{\eta_A}_{1/\alpha}=1} \text{Tr}_{\mathcal{H}_A} ( \eta_A \sigma_A^{1/2} \bar \tau_A^{-\alpha} \sigma_A^{1/2} ) = \sup_{\eta_A | \norm{\eta_A}_{1/\alpha}=1} \text{Tr}_{\mathcal{H}_A} ( \eta_A X_A^\dagger X_A ) = \norm{X_A^\dagger X_A}_{\frac{1}{1-\alpha}} =\norm{X_A X_A^\dagger}_{\frac{1}{1-\alpha}} \\
= \left (Tr\left [ (\tau_A^{-\alpha/2} \sigma_A \tau_A^{-\alpha/2})^{\frac{1}{1-\alpha}} \right ] \right )^{1-\alpha } = \exp ( \alpha S_\alpha(\sigma_A || \tau_A)),
\end{multline}
where 
$$S_\alpha(\sigma_A || \tau_A) =  \frac{1-\alpha}{\alpha}\log \left (Tr\left [ (\tau_A^{-\alpha/2} \sigma_A \tau_A^{-\alpha/2})^{\frac{1}{1-\alpha}} \right ] \right )$$
is the sandwiched relative Renyi entropy defined earlier, with a different index, $\alpha$, related to $n$ by $\alpha = \frac{n-1}{n}.$
Therefore,
$$S_\alpha( \Phi || \Omega : \mathcal{A} ) =\frac{1}{\alpha} \max_{\{p_A\}} \log \left( \sum_A q_A \frac{p_A^\alpha}{w_A^\alpha} \exp ( \alpha S_\alpha(\sigma_A || \tau_A)  )\right).$$
%$$S_\alpha( \Phi || \Omega : \mathcal{A} ) = \frac{1}{\alpha} \max_{\{p_A\}} \log \left( \sum_A q_A \frac{p_A^\alpha}{w_A^\alpha} \text{Tr}_{\mathcal{H}_A} ( \rho_A^{\alpha} \sigma_A^{1/2} \bar \tau_A^{-\alpha} \sigma_A^{1/2} ) \right )=\frac{1}{\alpha} \max_{\{p_A\}} \log \left( \sum_A q_A \frac{p_A^\alpha}{w_A^\alpha} \exp ( \alpha S_\alpha(\sigma_A || \tau_A)  )\right).$$
We need to maximize the quantity in the logarithm, subject to the constraint that $\sum_A p_A=1$. To do this, we introduce a Lagrange multiplier $\lambda$ to enforce the constraint. We then define 
$$f(p_A,\lambda) \equiv \sum_A q_A \frac{p_A^\alpha}{w_A^\alpha} \exp ( \alpha S_\alpha(\sigma_A || \tau_A)  ) - \lambda \left ( \sum_A p_A -1 \right ).$$
The maximum will then be given by the solution to the system 
$$\frac{\partial f}{\partial p_A} =0, \sum_A p_A =1.$$
This gives
$$q_A \frac{p_A^{\alpha-1}}{w_A^\alpha}\exp ( \alpha S_\alpha(\sigma_A || \tau_A)  )  = \lambda, \sum_A p_A =1,$$
which means
$$p_A^{\alpha-1}  = \lambda \frac{w_A^\alpha}{q_A} \exp (- \alpha S_\alpha(\sigma_A || \tau_A)  ),$$
$$p_A^{\alpha}  = \lambda^{ \frac{\alpha}{\alpha-1}} \frac{w_A^{ \frac{\alpha^2}{\alpha-1}}}{q_A^{ \frac{\alpha}{\alpha-1}}} \exp (- \frac{\alpha^2}{\alpha-1} S_\alpha(\sigma_A || \tau_A)  ).$$
Therefore, 
\begin{multline}
S_\alpha( \Phi || \Omega : \mathcal{A} ) = \frac{1}{\alpha}  \log \left( \sum_A q_A \lambda^{ \frac{\alpha}{\alpha-1}} \frac{w_A^{ \frac{\alpha^2}{\alpha-1}}}{q_A^{ \frac{\alpha}{\alpha-1}}} \exp (- \frac{\alpha^2}{\alpha-1} S_\alpha(\sigma_A || \tau_A)  ) \frac{1}{w_A^\alpha} \exp ( \alpha S_\alpha(\sigma_A || \tau_A)  )\right) \\
= \frac{1}{\alpha} \log \left( \sum_A q_A  \frac{w_A^{ \frac{\alpha^2}{\alpha-1}}}{q_A^{ \frac{\alpha}{\alpha-1}}} \exp (- \frac{\alpha^2}{\alpha-1} S_\alpha(\sigma_A || \tau_A)  ) \frac{1}{w_A^\alpha} \exp ( \alpha S_\alpha(\sigma_A || \tau_A)  )\right)+ \frac{1}{\alpha} \log \left ( \lambda^{\frac{\alpha}{\alpha-1}} \right).
\end{multline}
Let us begin by considering the first term, which we call $\bar S_\alpha( \Phi || \Omega : \mathcal{A} ).$ We have 
\begin{multline}
\bar S_\alpha( \Phi || \Omega : \mathcal{A} ) = \frac{1}{\alpha} \log \left( \sum_A q_A  \frac{w_A^{ \frac{\alpha^2}{\alpha-1}}}{q_A^{ \frac{\alpha}{\alpha-1}}} \exp (- \frac{\alpha^2}{\alpha-1} S_\alpha(\sigma_A || \tau_A)  ) \frac{1}{w_A^\alpha} \exp ( \alpha S_\alpha(\sigma_A || \tau_A)  )\right) \\
 = \frac{1}{\alpha} \log \left( \sum_A q_A^{\frac{\alpha-1}{\alpha-1}-\frac{\alpha}{\alpha-1}}  w_A^{ \frac{\alpha^2}{\alpha-1}-\frac{(\alpha^2-\alpha)}{\alpha-1}}   \exp ( \left[\frac{\alpha^2-\alpha}{\alpha-1} - \frac{\alpha^2}{\alpha-1}  \right]S_\alpha(\sigma_A || \tau_A)  )\right) \\
  = \frac{1}{\alpha} \log \left( \sum_A q_A^{-\frac{1}{\alpha-1}}  w_A^{ \frac{\alpha}{\alpha-1}}   \exp (\frac{-\alpha}{\alpha-1} S_\alpha(\sigma_A || \tau_A)  )\right). 
\end{multline}
Next, we must solve for $\lambda$. We know that 
$$p_A  = \lambda^{ \frac{1}{\alpha-1}} \frac{w_A^{ \frac{\alpha}{\alpha-1}}}{q_A^{ \frac{1}{\alpha-1}}} \exp (- \frac{\alpha}{\alpha-1} S_\alpha(\sigma_A || \tau_A)  ),$$
$$\sum_A p_A =1.$$
Hence,
$$\lambda^{ \frac{1}{\alpha-1}} \sum_A \left[ \frac{w_A^{ \frac{\alpha}{\alpha-1}}}{q_A^{ \frac{1}{\alpha-1}}} \exp (- \frac{\alpha}{\alpha-1} S_\alpha(\sigma_A || \tau_A)  ) \right ] =1$$
so that 
$$\lambda^{ \frac{1}{\alpha-1}}  =\frac{1}{\sum_A \left[ \frac{w_A^{ \frac{\alpha}{\alpha-1}}}{q_A^{ \frac{1}{\alpha-1}}} \exp (- \frac{\alpha}{\alpha-1} S_\alpha(\sigma_A || \tau_A)  ) \right ]}.$$
The second term in $S_\alpha( \Phi || \Omega : \mathcal{A} )$ therefore becomes
$$\frac{1}{\alpha} \log \left ( \lambda^{\frac{\alpha}{\alpha-1}} \right) =  \log \left ( \lambda^{\frac{1}{\alpha-1}} \right) = -\log \left ( \sum_A  q_A^{ \frac{1}{\alpha-1}}w_A^{ \frac{\alpha}{\alpha-1}} \exp (- \frac{\alpha}{\alpha-1} S_\alpha(\sigma_A || \tau_A)  ) \right).$$
Finally, putting everything together, we obtain an expression for the algebraic relative Renyi entropy, 
\begin{multline}
S_\alpha( \Phi || \Omega : \mathcal{A} )=\bar S_\alpha( \Phi || \Omega : \mathcal{A} )+\frac{1}{\alpha} \log \left ( \lambda^{\frac{\alpha}{\alpha-1}} \right) \\
 = \frac{1}{\alpha} \log \left( \sum_A q_A^{-\frac{1}{\alpha-1}}  w_A^{ \frac{\alpha}{\alpha-1}}   \exp (\frac{-\alpha}{\alpha-1} S_\alpha(\sigma_A || \tau_A)  )\right)-\log \left ( \sum_A  q_A^{ \frac{1}{\alpha-1}}w_A^{ \frac{\alpha}{\alpha-1}} \exp (- \frac{\alpha}{\alpha-1} S_\alpha(\sigma_A || \tau_A)  ) \right) \\
 =\frac{1-\alpha}{\alpha} \log \left( \sum_A q_A^{-\frac{1}{\alpha-1}}  w_A^{ \frac{\alpha}{\alpha-1}}   \exp (\frac{-\alpha}{\alpha-1} S_\alpha(\sigma_A || \tau_A)  )\right).
\end{multline}
We now rewrite this in terms of the index $n$, related to $\alpha$ by $\alpha = \frac{n-1}{n} = 1-\frac{1}{n}$ so that $n=\frac{1}{1-\alpha}$ and $\frac{\alpha}{1-\alpha} = n-1.$ Thus, we find
$$
S_n( \Phi || \Omega : \mathcal{A} )
 =\frac{1}{n-1} \log \left( \sum_A q_A^{n}  w_A^{ 1-n}   \exp ((n-1)S_n(\sigma_A || \tau_A)  )\right),
$$
which is exactly the form of relative Renyi entropy obtained in the holographic error-correction setting above. 

\section{Numerical Tensor Network Calculations}

\begin{figure}
\includegraphics[width=\textwidth]{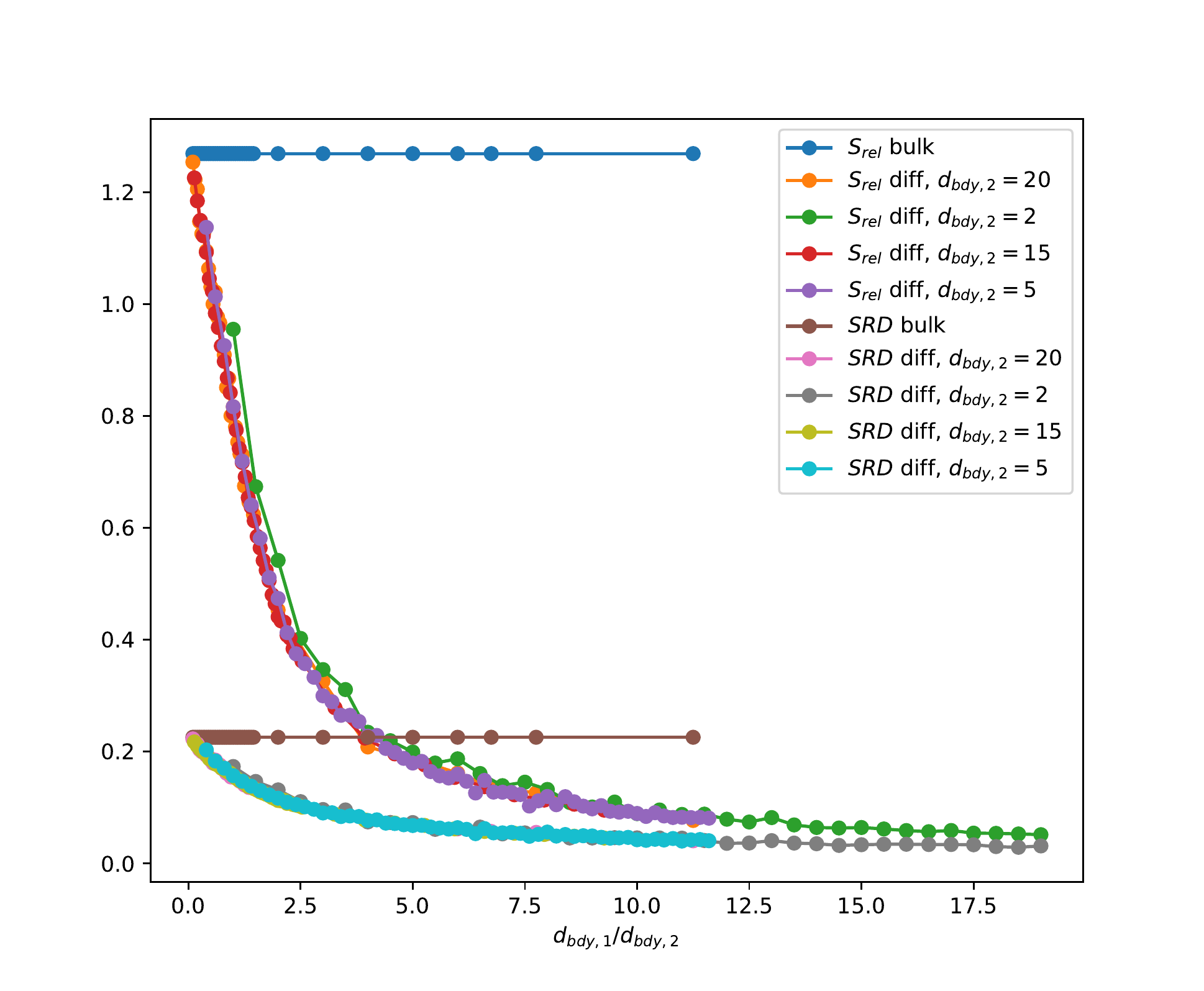}
\caption{}
\label{fig:S_rel_SRD_diff_ratio}           
\end{figure}

The discussion above established an equivalence between \textit{exact} equality of bulk and boundary sandwiched Renyi relative entropy and other entries in the AdS/CFT dictionary, such as the exact equality of bulk and boundary relative entropy as well as exact quantum error correction. Previous work has shown that approximate equality of bulk and boundary relative entropy can lead to bounds on approximate error correction~\cite{Approx}. Thus, it is of interest to have an understanding of approximate equality of bulk and boundary Renyi relative entropy. 

To this end, we consider random tensor networks as a model of AdS/CFT~\cite{RTN} (i.e., tensor networks where the tensors are drawn from probability distributions).  Random tensor networks have been a very useful tool for understanding features of the entanglement structure of AdS/CFT, such as the RT formula, the equality of bulk and boundary relative entropy, reflected entropy and so on~\cite{KFR,RTN, ReflectedEntropy}.  

We study a simple random tensor network with 1 tensor, 1 bulk qudit, whose Hilbert space has dimension $d_{bulk},$ and two boundary qudits, whose Hilbert spaces have dimensions $d_{bdy,1}$ and $d_{bdy,2}$. We consider two states $\rho_b, \sigma_b$ on the bulk qudit. We can calculate, for example, the bulk relative entropy between these states as usual. To find the corresponding boundary quantity, we proceed as follows. Let $\ket{0_x}$ be a fixed state on all the qudits. We then generate a Haar-random unitary $U$, and let $\ket{V_x} = U \ket{0_x}$. We then calculate the corresponding boundary state corresponding to our bulk state $\rho_b$:
$$\rho_{bdy} = Tr_{bulk} ( \rho_{bulk} \ket{V_x} \bra{V_x} ),$$
where $Tr_{bulk}$ is a trace over the bulk leg. The state on, e.g., boundary qudit 1 can obtained by taking a further partial trace $\rho_1 = Tr_2 \rho_{bdy}.$ To obtain the sandwiched Renyi divergence on boundary qudit 1, we compute the SRD of $\rho_1$ and $\sigma_1$, and average over Haar-random unitary matrices $U$.

We now consider specific cases, setting $d_{bulk} = 2$ throughout. First, we calculate the difference between the bulk relative entropy and the (averaged) relative entropy of the states on boundary qudit 1 for various values of the boundary bond dimensions, and the difference between the bulk SRD and the (averaged) SRD on boundary qudit 1 for Renyi index $\alpha = 0.2$. The results are shown in Figure~\ref{fig:S_rel_SRD_diff_ratio}. From this figure, we see that these differences only depend on the ratio of our boundary bond dimensions $x \equiv \frac{d_{bdy,1}}{d_{bdy,2}}.$ As we increase $x$, the difference converges to 0, so that equality of bulk and boundary relative entropy (as well as bulk and boundary SRD) holds.  

In Figure~\ref{fig:SRD_diff_dbdy2_20_diff_alpha}, we plot the SRD difference as a function of $d_{bdy,1}$ (with $d_{bdy,2}=20$) for two different values of $\alpha=0.8,7.4.$ The $\alpha=7.4$ difference seems to fall off somewhat faster.

\begin{figure}
\includegraphics[width=\textwidth]{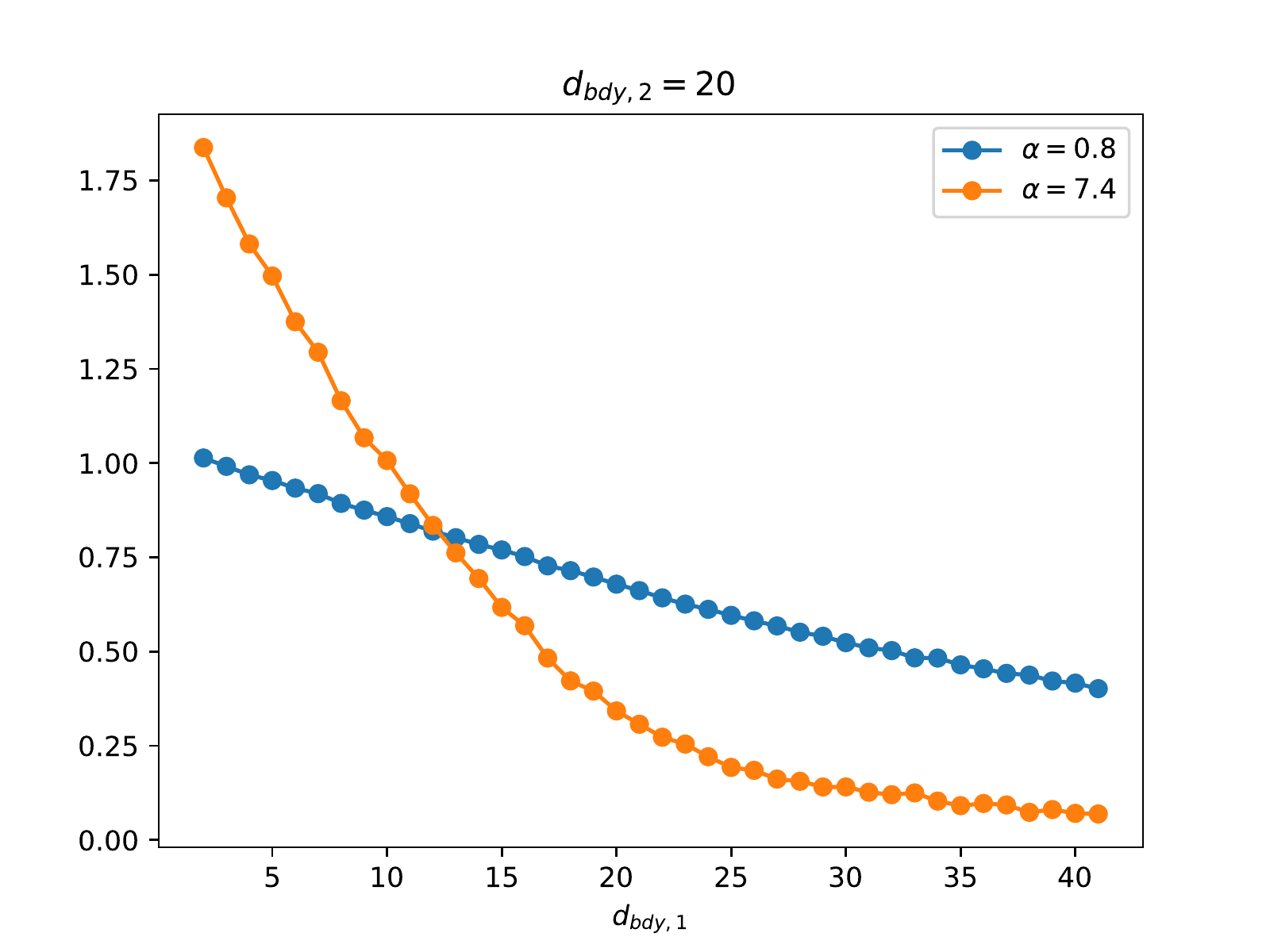}
\caption{}
\label{fig:SRD_diff_dbdy2_20_diff_alpha}           
\end{figure}

In Figure ~\ref{fig:SRD_alpha_d_bdy1_6_d_bdy2_2}, we plot the bulk and boundary qudit 1 SRD (and their difference) as a function of Renyi index $\alpha$, with $d_{bdy,1}=6, d_{bdy,2}=2$. In addition, we plot the data from Figure~\ref{fig:S_rel_SRD_diff_ratio} on a log-log plot, shown in Figure~\ref{fig:SRD_S_rel_log_plot}. 
It seems that the SRD difference is has a peak close to (but not equal to) $\alpha = 1$. At $\alpha=1$, of course, the SRD is the usual relative entropy.  Thus, there are values of the index $\alpha$ where the difference in SRD is smaller than the difference in relative entropy. It would be very interesting to see if these leads to new bounds on, for example, approximate reconstruction of bulk operators. We leave such investigations for future work.

%Interestingly, it seems that the SRD difference is has a peak very close to $\alpha = 1$, which of course is the value at which SRD is the usual relative entropy.

\begin{figure}
\includegraphics[width=\textwidth]{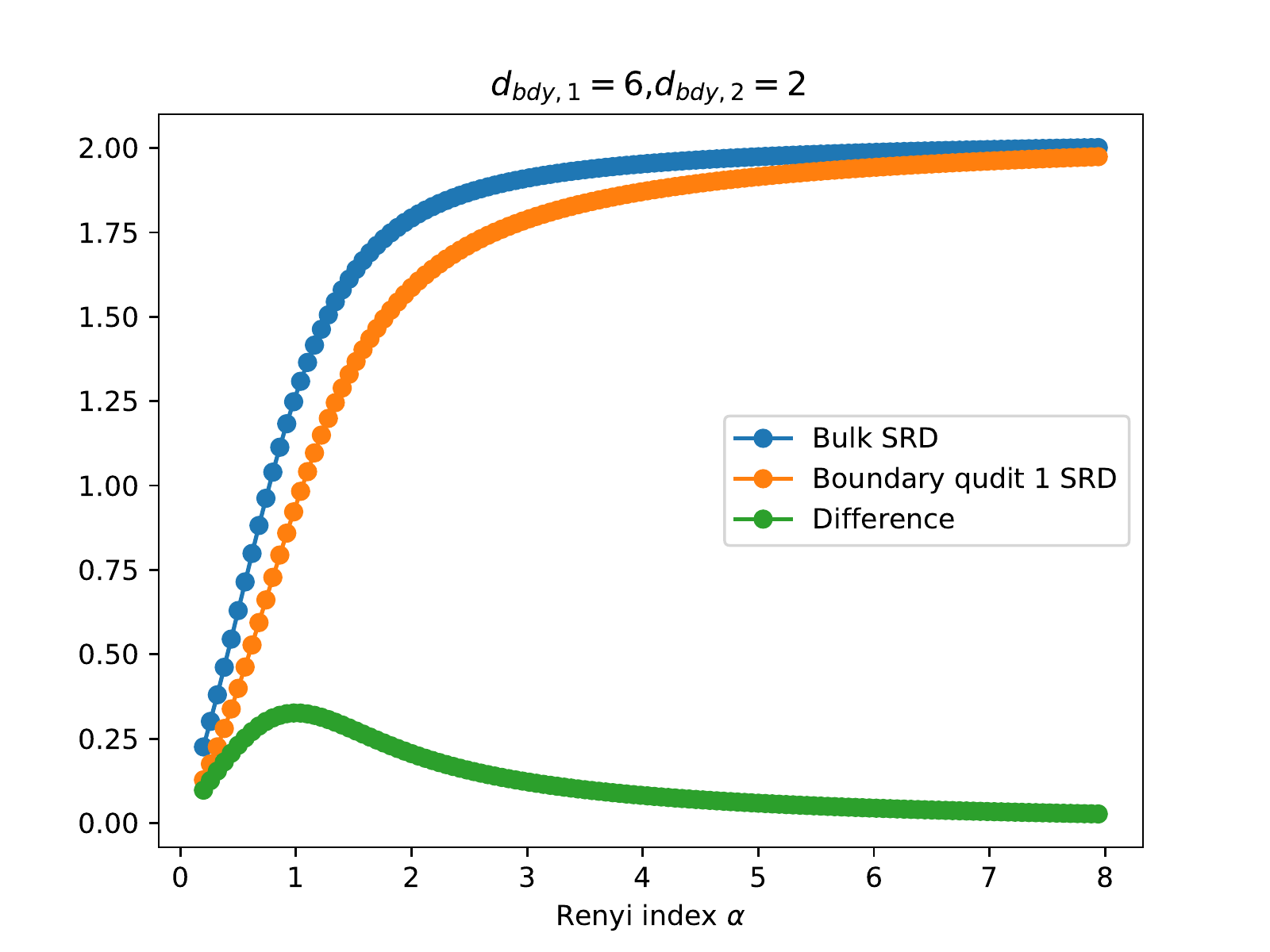}
\caption{}
\label{fig:SRD_alpha_d_bdy1_6_d_bdy2_2}           
\end{figure}

%We see that, during the transition, the entropy differences are approximately linear on the log-log scale, so that it obeys a power law in this regime, $\Delta S \sim x^p$. 

%We use a linear regression to find this exponent as a function of the Renyi index $\alpha$, and show the results in Figure~\ref{fig:SRD_power_law_scaling}. 

\begin{figure}
\includegraphics[width=\textwidth]{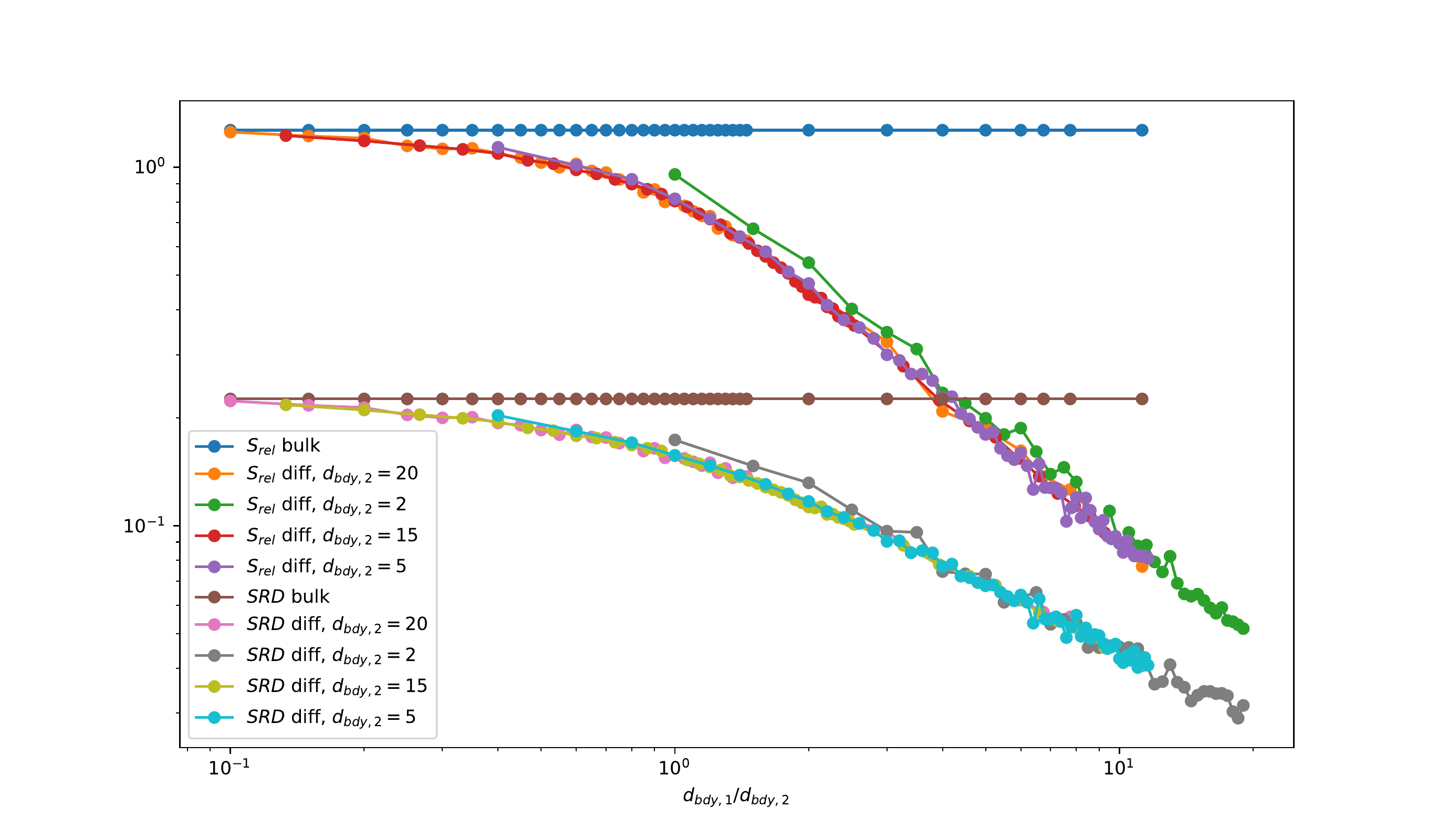}
\caption{}
\label{fig:SRD_S_rel_log_plot}           
\end{figure}

%\begin{figure}
%\includegraphics[width=\textwidth]{SRD_power_law_scaling.pdf}
%\caption{}
%\label{fig:SRD_power_law_scaling}           
%\end{figure}

%\section*{Path Integral Argument?}
%
%
%
%It would be interesting to see if one could somehow argue for equality of bulk and boundary sandwiched Renyi relative entropy in AdS/CFT directly from the path integral for states with semi-classical holographic duals. The sandwiched Renyi relative entropy is 
%
%$$S_n ( \rho || \sigma) \equiv \frac{1}{n-1} \log \left [ \text{Tr} \left[ \left ( \sigma^{\frac{1-n}{2n} } \rho \sigma^{\frac{1-n}{2n} }  \right )^n\right] \right ].$$
%Perhaps one could start by considering the related quantity 
%$$S_{n,m} ( \rho || \sigma) \equiv \frac{1}{n-1} \log \left [ \text{Tr} \left[ \left ( \sigma^{m } \rho \sigma^{m }  \right )^n\right] \right ]$$
%with $n,m \in \mathbf{N}.$ One can calculate this quantity on a replica manifold with $n$ copies of the following: $m$ copies of the path integral that prepares $\sigma$, a copy of the path integral that prepares $\rho$, $m$ copies of the path integral. See Figure~\ref{fig:replica_fig}. Perhaps one could show that $S_{n,m} (\rho || \sigma)$ are equal for the bulk and boundary for integer $n,m$, and use Carlson's theorem to show that they must be equal for non-integer $n,m$, and in particular for $m=  \frac{1-n}{2n}$ which would establish the result for sandwiched relative Renyi entropies.

%\appendix
%\section{Some title}
%Please always give a title also for appendices.
%

\section{Conclusions} 

In this paper, we have explored the role of the sandwiched Renyi relative entropy in AdS/CFT.  In particular, we have shown that the equivalence of bulk and boundary sandwiched Renyi relative entropy is equivalent to the RT formula, algebraic encoding, subregion duality, and the equivalence of bulk and boundary relative entropy, expanding the equivalence theorem (established in ~\cite{HarlowQEC}) of the latter four statements. We then discussed the Renyi relative entropies from the perspective of modular operators. In the context of finite-dimensional von Neumann algebras, this algebraic definition of the sandwiched Renyi relative entropy was shown to reduce to the form found in the context of the holographic error-correction setting. Finally, we explored numerical calculations of the sandwiched Renyi relative entropy in a simple random holographic tensor network.

There are several possible avenues for further investigation. Previous work has shown that the corrections to the equality of bulk and boundary relative entropy can bound errors on the reconstruction of low-energy bulk operators~\cite{Approx} using the twirled Petz map. It would interesting to see if the corrections to the sandwiched Renyi relative entropies can place a similar bound on the accuracy of the reconstruction of low-energy bulk operators, perhaps using some channel other than the twirled Petz map. In particular, it would be of great interest if there were situations in AdS/CFT where the bounds on the differences between bulk and boundary SRD's lead to a more accurate reconstruction of bulk operators than the bound on the difference between bulk and boundary relative entropies.  It would also be interesting to continue to do numerical simulations on larger tensor networks, in addition to investigating possible analytic results on sandwiched Renyi relative entropies in various tensor network models of holography.

\acknowledgments

I am deeply grateful to Jonah Kudler-Flam, Mudassir Moosa, and Pratik Rath for collaboration during the early stages of this project, as well as for fruitful discussions. I would also like to thank Yasunori Nomura for useful conversations and for his continued guidance. Helpful comments on this manuscript from Pratik Rath are gratefully acknowledged.  I am supported by the U.S. Department of Energy under grant Contract Number DE-SC0019380.

%\paragraph{Note added.} This is also a good position for notes added
%after the paper has been written.


\end{document}